\documentclass{article}
\usepackage{ijcai15}

\usepackage{times,xspace,url}
\usepackage{amsthm}

 \usepackage{graphics,graphicx}
 \usepackage{amssymb}

\usepackage[ruled,linesnumbered,noend,oldcommands]{algorithm2e}

\SetAlFnt{\small}
\SetAlCapFnt{\small}
\SetAlCapNameFnt{\small}
\SetAlCapHSkip{0pt}
\IncMargin{-\parindent}
\usepackage{amssymb}
\usepackage{amsmath}

\newcommand{\prob}{\ensuremath{\mathcal{P}}\xspace}
\newcommand{\D}{\ensuremath{\mathit{D}}\xspace}
\renewcommand{\P}{\ensuremath{\mathit{P}}\xspace}
\newcommand{\C}{\ensuremath{\mathit C}\xspace}
\renewcommand{\S}{\ensuremath{\mathit S}\xspace}
\newcommand{\I}{\ensuremath{\mathcal{I}\xspace}}

\newcommand{\Dplus}{\ensuremath{\mathit{D}^+}\xspace}
\newcommand{\Pplus}{\ensuremath{\mathit{P}^+}\xspace}
\newcommand{\nomatch}{\ensuremath{\mathit{nil}}\xspace}
\newcommand{\rol}[1]{\ensuremath{\mathit{rol}_{#1}}\xspace}
\newcommand{\rols}{\ensuremath{\mathit{rols}}\xspace}

\newcommand{\pgt}[1]{\succ_{#1}}
\newcommand{\pgeq}[1]{\succeq_{#1}}

\newcommand{\pleq}[1]{\preceq_{#1}}

\newcommand{\mgt}{\succ_{\cal R}}
\newcommand{\mgeq}{\succeq_{\cal R}}
\newcommand{\mlt}{\prec_{\cal R}}
\newcommand{\mleq}{\preceq_{\cal R}}
\newcommand{\Ropt}{\mathcal{R}_{\mathit{opt}}}
\newcommand{\RPopt}{\mathcal{RP}_{\mathit{opt}}}

\newcommand{\quota}[1]{\ensuremath{\mathit{cap_{#1}}}}
\newcommand{\match}{\ensuremath{\mu}}
\newcommand{\matchinv}{\ensuremath{\mu^{-1}}}
\newcommand{\matchDvar}[2]{\ensuremath{m_{#1}[#2]}}
\newcommand{\matchCvar}[2]{\ensuremath{m_{#1}[#2]}}
\newcommand{\matchPvar}[3]{\ensuremath{m_{#1}[#2, #3]}}
\newcommand{\willaccept}{\ensuremath{\mathit{willAccept}}}
\newcommand{\choice}[1]{\ensuremath{\mathit{Ch}_{#1}}}
\newcommand{\ranked}[1]{\ensuremath{\mathit{ranked(#1)}}}
\newcommand{\rank}[1]{\ensuremath{\mathit{rank}_{#1}}}

\newcommand{\true}{\ensuremath{\mathit{true}}}

\newcommand{\smp}{\mbox{SMP}\xspace}
\newcommand{\smpc}{\mbox{SMP-C}\xspace}

\newcommand{\scite}{\shortcite}

\newtheorem{thm}{Theorem}{\bfseries}{\rmfamily}
{\bfseries}{\rmfamily}
{\bfseries}{\rmfamily}
{\bfseries}{\rmfamily}
{\bfseries}{\rmfamily}
{\bfseries}{\rmfamily}
{\bfseries}{\rmfamily}
\newtheorem{defn}{Definition}{\bfseries}{\rmfamily}

\title{Exploring Strategy-Proofness, Uniqueness, and Pareto Optimality for the Stable Matching Problem with Couples}
\date{\today}
\author{Andrew Perrault \and Joanna Drummond \and Fahiem Bacchus\\
Department of Computer Science\\
University of Toronto\\
\{perrault,jdrummond,fbacchus\}@cs.toronto.edu
}

\begin{document}

\maketitle

\begin{abstract}
  The Stable Matching Problem with Couples (\smpc) is a ubiquitous
  real-world extension of the stable matching problem (\smp) involving
  complementarities. Although \smp can be solved in
  polynomial time, \smpc is NP-Complete.
  Hence, it is not clear which, if any, of the theoretical results
  surrounding the canonical \smp
  problem apply in this setting. In this paper, we
  use a recently-developed SAT encoding to solve
  \smpc exactly. This allows us to enumerate all
  stable matchings for any given instance of \smpc. With this tool, we
  empirically evaluate some of the properties that have been
  hypothesized to hold for \smpc.

  We take particular interest in investigating if, as the size of the
  market grows, the percentage of instances with unique stable
  matchings also grows. While we did not find this trend among 
  the random problem instances we sampled, we did find that the
  percentage of instances with an resident optimal matching 
  seems to more closely follow the trends
  predicted by previous conjectures. We also define and investigate
  resident Pareto optimal stable matchings, finding that, even though
  this is important desideratum for the deferred acceptance style
  algorithms previously designed to solve \smpc, they do not always
  find one.

  We also
  investigate strategy-proofness for \smpc, showing that even if only
  one stable matching exists, residents still have incentive to
  misreport their preferences. However, if a problem has a
  resident optimal stable matching, we show that residents
  cannot manipulate via truncation.
\end{abstract}





\section{Introduction}
Stable matching problems are ubiquitous, with many real-world
applications, ranging from dating markets, to labour markets, to the
school choice problem. Arguably the most well-known of these
applications is the residency matching problem
\cite{niederle_roth_sonmez:2008,roth:JPE1984,somnez-roth:bostonAER2005},
instantiated in the large National Resident Matching Program (NRMP),
among others. Importantly, these markets have adapted to the needs of
their participants over the years, allowing them to more richly
express their preferences, or guarantee certain properties of the
match valuable to participants \cite{nrmp_data:2013}. To address the
problem of couples wanting to coordinate their placements, the NRMP
began allowing couples to jointly express their preferences over
residency programs. We call the resulting matching problem the Stable
Matching Problem with Couples (\smpc).

As with the standard \smp, the goal of \smpc is to find a matching
such that no pair (one from each side of the market) has an incentive
to defect from their assigned placement. Such a matching is called
\emph{stable}. In their seminal paper on stable matching, Gale and
Shapley provide a polytime Deferred Acceptance (DA) algorithm
guaranteed to find a stable, male-optimal (resident-optimal) matching
\scite{gale62}. They also proved many useful properties of \smp; e.g.,
a unique stable matching always exists, and a stable resident-optimal
matching always exists (a matching is resident-optimal if no resident is
better off in any other stable matching).

However, when a pair of doctors are allowed to jointly specify their
rankings, many of these desirable properties no longer hold. In
particular, \smpc is NP-complete \cite{ronn1990np}. Also, a stable
matching might not exist. Even when it does exist a resident-optimal
matching might not exist. \smpc is not the only NP-complete extension
of \smp. Previous literature on these other NP-complete extensions has
investigated techniques that are not guaranteed to return stable
matchings, such as local search techniques (e.g.,
\cite{marx2011stable,gelain2010local}). Another approach has entailed
encoding these NP-complete stable matching problems (e.g., the stable
matching problem with incomplete lists and ties) into other
NP-complete formalisms such as constraint satisfaction problems (CSP)
or satisfiability (SAT) (e.g.,
\cite{gent2002sat,gent2001constraint,DBLP:conf/sara/UnsworthP05}). The
advantage of such encodings is that effective solvers exist for CSP
and SAT. However, these works do not address \smpc.

Previous work on \smpc has mainly focused on providing sound but not
complete algorithms that extend the DA algorithm to deal with
couples \cite{roth99,kpr13}.
A new SAT encoding for \smpc was recently developed \cite{dpb_review},
showing very promising results when run on a competitive SAT solver
\cite{biere2013lingeling}. This encoding, SAT-E, scales fairly well,
and is sometimes able find stable matchings for problems on which the
DA-style algorithms cannot find any. Importantly, this new SAT
encoding allows for flexibility unavailable in DA-style algorithms:
adding new constraints and objectives can be as simple as adding
another constraint to the encoding.

In this paper, we use this flexibility to allow us to enumerate all
possible stable matches of an instance of \smpc.
We also introduce a constraint that allows us to, given a stable
solution to \smpc, determine if that matching is resident Pareto
optimal, and if it is not, find a resident Pareto optimal matching. 
This new machinery allows us to empirically evaluate interesting
properties of \smpc. In particular, we are interested in the following
properties: when do instances of \smpc have a unique matching? If an
instance has multiple stable matchings, when is there a resident
optimal matching? How does the existence of a stable matching change
as the percentage of couples in the market increases?

Previous theoretical work has addressed some of these questions;
Ashlagi \emph{et al.}~and Kojima \emph{et al.}~proved theoretical
results regarding the guaranteed existence of a stable matching, if
the size of the market grows sufficiently faster than the percentage
of couples in the matching \scite{braverman14,kpr13}. Immorlica and
Mahdian proved that, for \smp, and a fixed length ROL (with everyone
drawing from some given distribution), the chance of drawing a problem
with a unique stable matching goes to 1 as the size of the market
goes to infinity. However, to our knowledge, no analogous results
exist for \smpc. We thus create our experiments in an attempt to
simulate the settings in these papers.

Also, previous work has conjectured that some of the properties
present in \smp may carry over to \smpc. Roth and Peranson
hypothesized that the reason that they see few opportunities
for strategic behaviour in their data is that there are few stable
matchings in large \smpc markets \scite{roth99}.
Furthermore, Roth and Peranson cite that strategy-proofness and
resident optimality are desirable properties for stable matchings.

In this paper, we first provide a theoretical contribution, showing
that, for \smpc, we are not guaranteed resident strategy-proofness,
even in a problem with a unique stable matching. We do show, however,
when a matching has resident optimal solution, residents cannot
manipulate via truncation; they must manipulate via reordering their
preferences. We thus show that, for \smpc, uniqueness and resident
optimality are both insufficient for strategy-proofness, but, in
certain situations, it is ``harder'' for the residents to manipulate.
We then provide two extensions to the already existing SAT-E encoding
for \smpc, allowing us to enumerate all stable matches and find a
resident Pareto optimal match. We use these tools to empirically
explore the space of \smpc solutions. Our experiments suggest that
Roth and Peranson's observation that there were few instances where
residents could improve by truncation is due to the high fraction of
instances that have a resident optimal matching, a fraction that
appears to increase with instance size.

\section{Background}
We cast our formalization of \smpc (stable matching problem with
couples) in terms of the well known residency matching
problem, a labour market where doctors are placed into hospital
residency programs \cite{roth99}.

Informally, doctors wish to be placed into (matched with) some
hospital program, and programs wish to accept some number of doctors.
Both doctors and hospitals have preferences over who they are matched
with; expressed as ranked order lists (ROLs). Some doctors are paired
into couples, and these couples provide a joint ROL specifying their
joint preferences. Both doctors and hospitals can provide incomplete
lists, any alternative not listed is considered to be unacceptable.
That is, they would rather not be matched at all than matched to an
alternative not on their ROL. The \smpc problem to find a
\emph{stable} matching, such that no doctor-hospital pair has an
incentive to defect from the assigned matching.

\subsection{\smpc}
More formally, let \D be a set of doctors and \P be a set of programs.
Since there is a preference to be unmatched over an unacceptable
match, we use $\nomatch$ to denote this ``doctor'' or ``program''
alternative: matching a program $p$ to $\nomatch$ indicates that $p$
has an unfilled slot while matching a doctor $d$ to $\nomatch$
indicates that $d$ was not placed into any program. We use $\Dplus$
and $\Pplus$ to denote the sets $\D\cup \{\nomatch\}$ and
$\P \cup \{\nomatch\}$ respectively. 

The doctors are divided into two disjoint subsets, $\S \subseteq \D$
and $\D\setminus \S$. $\S$ is the set of single doctors and
$\D\setminus\S$ is the set of doctors that are in couple
relationships. We specify the couples by a set of pairs
$\C \subseteq (\D\setminus \S) \times (\D\setminus \S)$. If
$(d_1, d_2) \in C$ we say that $d_1$ and $d_2$ are each other's
partner. We require that every doctor who is not single (i.e., every
doctor in $\D\setminus \S$) have one and only one partner in $\C$.

Each program $p\in\P$ has an integer quota $\quota{p} > 0$. This quota
determines the maximum number of doctors $p$ can accept (i.e., $p$'s
capacity).
 
Everyone participating in the matching market has preferences over
their alternatives. Each participant $a$ specifies their preferences
in a ROL which lists $a$'s preferred matches from most preferred to
least preferred. The ROLs of single doctors $d\in \S$ contain programs
from $\Pplus$; the ROLs of couples $c \in \C$ contain pairs of
programs from $\Pplus \times \Pplus$; and the ROLs of programs
$p\in \P$ contain doctors from $\Dplus$. Every ROL is terminated by
$\nomatch$ (couple ROLS are terminated by $(\nomatch, \nomatch)$)
since being unmatched is always the least preferred option, but is
preferred to any option not on participant's $a$'s ROL. 

The order of items on $a$'s ROL defines a partial ordering relation
where $x \pgeq{a} y$ indicates that $x$ appears before $y$ on $a$'s
ROL or $x=y$. We define $\pgt{a}$, $\pleq{a}$, and $\pgt{a}$ in terms
of $\pgeq{a}$ and equality in the standard way. We say that $x$ is
\emph{acceptable} to $a$ if $x \pgeq{a} \nomatch$. (Note that
unacceptable matches are not ordered by $\pgeq{a}$.)

We define a choice function $\choice{p}()$ for programs $p\in \P$.
Given a set of doctors $R$, $\choice{p}(R)$ returns the subset of $R$
that $p$ would prefer to accept. $\choice{p}(R)$ is the maximal subset
of $R$ such that for all $d\in \choice{p}(R)$, $d\pgt{p}\nomatch$, for
all $d'\in R{-}\choice{p}(R)$, $d\pgt{p}d'$, and
$|\choice{p}(R)| \leq \quota{p}$. It is convenient to give the null
program a choice function as well: $\choice{\nomatch}(O) = O$, i.e.,
$\nomatch$ will accept any and all matches.

We use the notation $\ranked{a}$ to denote the set of options that $a$
could potentially be matched with. For single doctors $d$ and programs
$p$ this is simply the ROLs of $d$ ($\rol{d}$) and $p$ ($\rol{p}$).
For a doctor that is part of a couple $(d_1,d_2)$,
$\ranked{d_1} = \{p_1 | \exists p_2. (p_1,p_2)\in \rol{(d_1,d_2)}\}
\cup \{\nomatch\}$
and similarly for $d_2$. Note that $\nomatch\in \ranked{a}$.

Finally, we use the function $\rank{a}(x)$ to find the index of match
$x$ in $a$'s \rol{}: $\rank{a}(x) = i$ iff $x$ appears at index $i$
(zero-based) on $a$'s ROL, $\rol{a}$, or $|\rol{a}|$ if $x$ does not
appear on $\rol{a}$. Also we use $\rol{a}$ as an indexable vector,
e.g., if $x$ is acceptable to $a$ then $\rol{a}[\rank{a}(x)] = x$.

\subsection{Stable Matchings}
\begin{defn}[Matching] A \textbf{matching} $\match$ is a mapping from
$\D$ to $\Pplus$. We say that a doctor $d$ is matched to a program $p$
under $\match$ if $\match(d) = p$, and that $p$ is matched to $d$ if
$d \in \matchinv(p)$.
\end{defn}

We want to find a matching where no doctor-program pair has an
incentive to defect. We call the pairs that do have an incentive to
defect blocking pairs. First we define the condition
$\willaccept(p,R,\match)$ to mean that program $p$ would prefer to
accept a set of doctors $R$ over its current match $\matchinv(p)$: 
$\willaccept(p,R,\match) \equiv R\subseteq \choice{p}(\matchinv(p)
\cup R)$.
Note that if $p$ is already matched to $R$ under $\match$ then $p$
will accept $R$.

\begin{defn}[Blocking Pairs for a Matching $\match$] 
\begin{enumerate}
\item 
  A single doctor $d\in \S$ and a program $p\in \P$ is a
  \textbf{blocking pair} for
  $\match$ if and only if $p \pgt{d} \mu(d)$ and
  $\willaccept(p,\{d\},\match)$ 
\item 
  A couple $c=(d_1,d_2)\in \C$ and a program pair
  $(p_1,p_2)\in \Pplus\times
  \Pplus$
  with $p_1\neq p_2$ is a \textbf{blocking pair} for $\match$ if and
  only if
  $(p_1,p_2) \pgt{(d_1,d_2)} (\match(d_1), \match(d_2))$, 
  $\willaccept(p_1,\{d_1\}, \match)$, and $\willaccept(p_2, \{d_2\},
  \match)$.
\item A couple $c=(d_1,d_2)$ and a program $p\in \P$ is a
  \textbf{blocking pair} for $\match$ if and only if
  $(p,p) \pgt{(d_1,d_2)}
  (\match(d_1), \match(d_2))$ and $\willaccept(p,\{d_1, d_2\},
  \match)$. 
\end{enumerate}
\end{defn}

\begin{defn}[Individually Rational Matching] A matching $\match$ is
  \textbf{individually rational} if and only if (a) for all $d\in\S$,
  $\match(d) \pgeq{d} \nomatch$, (b) for all $c=(d_1,d_2)\in \C$,
  $(\match(d_1), \match(d_2)) \pgeq{c} (\nomatch, \nomatch)$, and (c)
  for all $p\in \P$, $|\matchinv(p)| \leq \quota{p}$ and for all
  $d\in \matchinv(p)$ we have that $d \pgeq{p} \nomatch$.
\end{defn}

\begin{defn}[Stable Matching] A matching $\match$ is \textbf{stable}
  if and only if it is individually rational and no blocking pairs for
  it exist.
\end{defn}

\subsection{Residence Preferred Matchings}
The set of stable matchings can be quite large. When there are no
couples (i.e., in \smp) this set is always non-empty \cite{gale62} and
has a nice structure: it forms a lattice under the partial order
$\mgeq$ defined below \cite{knuth_smp}.

\begin{defn}[Residence Preferred Matchings]
  A matching $\match_1$ is resident preferred to another
  $\match_2$ when for all $a\in \S \cup \C$ we have that
  $\match_1(a) \pgeq{a} \match_2(a)$. We denote this
  relationship as $\match_1 \mgeq \match_2$. We also define $\mgt$,
  $\mlt$, and $\mleq$ in terms of $\mgeq$ and equality in the standard
  way. When $\match_1\mgt\match_2$ we say that $\match_1$
  \textbf{dominates} $\match_2$. 
\end{defn}
In other words, $\match_i \mgt \match_2$ if every doctor and couple
gets the same or a better match in $\match_1$ as in $\match_3$ and at
least one doctor or couple gets a better match.

The lattice structure and the fact there can only be a finite number
of matchings means that for \smp there is always a resident-optimal
matching.

\begin{defn}[Resident Optimal Matching] We say that a matching
  $\match$ is resident optimal, written $\Ropt(\match)$ when $\match$
  is stable and $\match\mgeq \match'$ for all other stable matchings
  $\match'$ ($\match$ dominates all other stable matchings).
\end{defn}
Note that, we restrict the resident-optimal matching to be stable. It
can be observed that for any resident $r$ (couple $(r_1,r_2)$)
$\match(r)$ ($(\match(r_1),\match(r_2)$) is the best match offered by
any stable matching when $\Ropt(\match)$.

Resident-optimality is generally cited as an important property for
stable matching algorithms (e.g., \cite{gale62,roth99}). In the
presence of couples however, stable matchings might not exist and
even when they do there have no lattice structure under $\mgeq$.
However, the $\mgeq$ is still well defined and for \smpc leads to
potentially multiple \emph{Pareto optimal} matchings.
\begin{defn}[Resident Pareto Optimal Matchings] We say that a matching
  $\match$ is resident Pareto optimal, written $\RPopt(\match)$ when
  $\match$ is stable and there does not exist another stable matching
  $\match'$ such that $\match' \mgt \match'$.
\end{defn}

It is easy to see that in \smpc every stable matching $\match$ either
is an $\RPopt$ or is dominated by an $\RPopt$ matching. This means
that an \smpc instance has an $\Ropt$ matching if and only if it has a
unique $\RPopt$ matching. As we will see in
Section~\ref{sec:empirical} it is often the case that \smpc instances
have more than one $\RPopt$ matching (and thus no $\Ropt$ matching).

\section{Strategy-Proofness and $\Ropt$ Existence for \smpc}
For an \smp instance the mechanism that returns the $\Ropt$ matching
in response to the stated preferences is strategy-proof with respect
to the residents. That is, it is a dominant strategy for each resident
to state their true preferences \cite{roth_two_sided}.

To our knowledge, no strategy-proofness results exist for \smpc.
However, this result for \smp leads us to hypothesize that if we
restrict our attention to \smp instances in which an $\Ropt$ matching
exists, then a mechanism that returns the $\Ropt$ matching is
strategy-proof. Unfortunately, this is not true; furthermore, the even
stronger condition of there being a unique stable matching for a
problem instance doesn't guarantee strategy-proofness for the
residents (let alone the hospitals). However, we do show that when
there is an $\Ropt$ matching, residents have no incentive to misreport
their preferences via truncation.

\begin{thm} Restricting to instances of \smpc in which an $\Ropt$
  matching exists, let $y()$ be a mechanism that maps from such an
  instance to the $\Ropt$ matching. Then $y$ is strategy-proof for
  residents if residents are only allowed to manipulate via
  truncations. However, $y$ might not be
  strategy-proof if residents are allowed to manipulate via
  reordering.
\end{thm} 

\begin{proof}
  Let $\mu^*$ be the $\Ropt$ matching.

  Part I: Truncations are strategy-proof. Suppose resident or couple
  $d$ truncates its preferences at program or pair of programs $p$.
  Let $\Omega$ be the set of stable matchings before the truncation
  and let $\Omega'$ be the set of stable matchings after the
  truncation. Because $d$'s preferences above the truncation point are
  the same after truncation, any matching $\match \in \Omega'$ where
  $\match(d) \neq \nomatch$ (resp.
  $\match(d) \neq (\nomatch, \nomatch)$ for a couple $d$) is in
  $\Omega$: the stability conditions contributed by
  other residents and couples are the same after $d$'s truncation.
  However, $\Omega'$ may contain stable matchings where $d$ is matched
  to $nil$ that are not in $\Omega$. Thus
  $\{ \match \in \Omega' : \match(d) \neq \nomatch \} \subseteq \Omega
  $.

  There are two cases: either 1) $d$ truncates below $\match^*(d)$ or
  2) $d$ truncates at or above $\match^*(d)$. In the former, it can be
  seen that $\mu^*$ is a stable matching of $\Omega'$ and the
  only
  matches in $\Omega'$ not in $\Omega$ have $d$ matched to $\nomatch$.
  Hence, $\mu^*$ is an $\Ropt$ matching for $\Omega'$ and $d$ cannot
  improve. In the
  latter, $\Omega'$ consists of
  matchings where $d$ is matched to $\nomatch$ because any
  $\match \in \Omega'$ such that $\match(d) \succ_d \match^*(d)$ would
  be in $\Omega$ and would contradict $\mu^*$'s $\Ropt$ status. Hence,
  $d$ cannot improve in this case either.

  Part II: Reordering is not strategy-proof. We provide a
  counterexample. See Figure \ref{fig:strategy-proof_counter}. The
  preferences provided above the line show doctors' and programs' true
  preferences. Couples are identified by their doctor-doctor pair,
  and give joint preferences as expected. Given these true
  preferences, only one stable matching exists (shown in the
  third column above the line). However, even though only one
  stable matching exists (and thus this matching is $\Ropt$), the
  single
  doctor $r_0$ has incentive to misreport preferences via reordering.
  Instead of reporting $a \succ b \succ c \succ d$ ($r_0$'s true
  preferences) $r_0$ can be matched to $b$ instead of $c$ by reporting
  $b \succ d \succ c \succ a$. (All other participants in the market
  report their true preferences.) This results in the matching seen in
  the third column below the line, a matching that is not stable under
  the original preferences, where single doctor $r_0$ is better off
  than when they reported their true preferences.
\end{proof}

\begin{small}
\begin{figure*}
\begin{align*}
&\textrm{Doctor preferences} && \textrm{Program preferences} & \textrm{Unique }&\textrm{stable matching} \\
r_0&: a \succ b \succ c \succ d & a&: r_3\succ r_0 \succ r_1 & \mu(r_0) &= c \\
(r_1,r_2)&: (b, e) \succ (a, d) & b&: r_1\succ r_0 & \mu((r_1,r_2)) &= (b,e) \\
(r_3,r_4)&: (a, d) \succ (c, e) & c&: r_3\succ r_0 & \mu((r_3,r_4)) &= (a,d) \\
&& d&: r_0\succ r_2\succ r_4& \\
&& e&: r_4\succ r_2& \\
\\
\hline
\\
&\textrm{Doctor preferences} && \textrm{Program preferences} & \textrm{Unique }&\textrm{stable matching} \\
r_0&: b \succ d \succ c \succ a & a&: r_3\succ r_0\succ r_1 & \mu(r_0) &= b \\
(r_1,r_2)&: (b, e) \succ (a, d) & b&: r_1\succ r_0 & \mu((r_1,r_2)) &= (a,d) \\
(r_3,r_4)&: (a, d) \succ (c, e) & c&: r_3\succ r_0 & \mu((r_3,r_4)) &= (c,e) \\
&& d&: r_0\succ r_2\succ r_4& \\
&& e&: r_4\succ r_2&
\end{align*}
\caption{Counterexample for Strategy-Proofness in the Reordering Case}
\label{fig:strategy-proof_counter}
\end{figure*}
\end{small}

Roth and Peranson investigate strategy-proofness for both residents
and hospitals empirically in the NRMP setting \scite{roth99}. They do
this via truncations. In their setting, they do find that residents
can benefit (albeit minimally) from truncating their true preferences.
We then hypothesize that this is due either: i)
the instances in their setting frequently do not contain an $\Ropt$ matching, or ii)
their algorithm (RP99) does not find the $\Ropt$ matching
even when it does exist. In Section~\ref{sec:empirical} we will see
that both of these situations can occur. That is, on many problems no
$\Ropt$ matching exists (the best one can do is an $\RPopt$ matching).
And even when an $\Ropt$ matching does exist, their algorithm might
not find it, reporting some other stable matching instead.

\section{Solving \smpc}
In this section we examine solution methods for finding stable
matchings for \smpc. The standard approach has been to find an
extension of the deferred acceptance algorithm that can handle
couples. However, these extensions are incomplete: they are unable to
determine whether or not a stable matching exists, and even when a
stable matching does exist they might not be able to find one.

Since \smpc is known to be NP-Complete \cite{ronn1990np} it is also
possible to encode it as another NP-Complete problem. In particular,
it can be encoded as a SAT (satisfiability) problem. The advantage of
doing this is that there has been a great deal of work on SAT solvers
and state-of-the-art SAT solvers are routinely able to solve SAT
problems involving millions of clauses.\footnote{Of course, since SAT
  is NP-Complete these solvers can also be foiled by small SAT
  instances. Nevertheless, they often exhibit very good performance
  on instances that arise from real world problems.} The advantage of
encoding \smpc to SAT and then using a SAT solver is that this
approach is complete: given sufficient compute resources it will
either find a stable matching or prove that none exists.

\subsection{Existing DA Algorithms for \smpc}
\label{sec:da}
The basic principle of DA algorithms \cite{gale62} is that members of
one side of the market propose down their ROLs while the other side
either rejects those proposals or holds them until they see a better
proposal: once all proposals have been made that side accepts the
proposal they have not rejected (acceptance of proposals is deferred
until the end).


Roth and Peranson develop a DA algorithm, RP99, capable of dealing
with couples \scite{roth99}. This well known algorithm has been used
with considerable success in practice, including most famously for
finding matches for the NRMP which typically involves about 30,000
doctors \cite{nrmp_data:2013}. Using the description in \cite{roth99}
we have implemented RP99. RP99 employs an iterative scheme. After
computing a stable matching for all single doctors, couples are added
one at a time and a new stable matching computed after each addition.
The algorithm uses DA at each stage to find these stable matchings.
Matching a couple can make previously made matches unstable and in
redoing these matches the algorithm might start to cycle. Hence, cycle
checking (or a timeout) is sometimes needed to terminate the algorithm.

Kojima et al. develop a simple ``sequential couples algorithm''
\scite{kpr13}.
This algorithm is analyzed to prove that the probability of a stable
matchings existing goes to one under certain assumptions. However,
this simple algorithm is not useful in practice as it declares failure
under very simple conditions. Kojima et al. also provide a more
practical DA algorithm, KPR, that they use in their experiments. We
have implemented KPR. The main difference between KPR and RP99 is that
KPR deals with all couples at the same time---it does not attempt to
compute intermediate stable matchings. Interestingly, this makes KPR
much more successful (and efficient) in our experiments.

Finally, Ashlagi et al. extend the analysis of Kojima et al.,
developing
a more sophisticated ``Sorted Deferred Acceptance Algorithm'' and
analyzing its behaviour \scite{braverman14}. This algorithm is
designed mainly to be amenable to theoretical analysis rather than
for practical application. We have not implemented this algorithm so
we do not include any empirical results about its performance.

\subsection{Solving \smp via SAT}
We have developed an effective encoding of \smpc into SAT. This
encoding and its performance is reported on in
\cite{dpb_review}.\footnote{For the reader's convenience we have
  included a description of this encoding in the
  appendix. See \cite{dpb_review} for full details.}

We call this encoding SAT-E, and given an \smpc instance
$\langle \D, \C, \P, \mathit{ROLs} \rangle$, where $\mathit{ROLs}$ is
the set of all participant ROLs,
SAT-E($\langle \D,\C, \P, \mathit{ROLs} \rangle$) can be viewed as a
function that returns a SAT encoding in CNF (conjunctive normal form),
which is the input format taken by modern SAT solvers.

There are three things to know about SAT-E.
\begin{enumerate}
\item For any \smpc instance
  $\I = \langle \D,\C, \P, \mathit{ROLs} \rangle$, the
  satisfying models of SAT-E stand in a one-to-one correspondence with
  the stable
  models of $\I$. This allows us to enumerate all stable
  models (see below).

\item SAT-E includes the set of propositional variables $m_d[p]$ one
  for each $d\in \D$ and $p\in \D$ such that $p\pgeq{d} \nomatch$. In
  any satisfying model, $\pi$, $m_d[p]$ is true if and only if
  $\mu(d)=p$ in the stable matching $\mu$ corresponding to $\pi$.

\item SAT-E also includes the set of propositional variables $m_c[i]$
  for each $c\in \C$ and $i$ in the range $[0,|\rol{c}-1|]$ where
  $\rol{c}$ is $c$'s ROL. In any satisfying model, $\pi$, $m_c[i]$ is
  true if and only if in $\mu$, the stable matching corresponding to
  $\pi$, $c$ is matched to a program pair they rank $i$
  or above in their ROL. 
\end{enumerate}

Given SAT-E and the above three facts we can develop two simple yet
powerful algorithms: one for enumerating all stable models and the
other for finding a Pareto Optimal matching, an $\RPopt$ matching,
that dominates a stable matching $\mu$.

\begin{algorithm}[t]
  \dontprintsemicolon
  \SetCommentSty{textsl}
  \SetNoLine
  \SetKwComment{Comment}{\textbf{/* }}{\textbf{*/}}
  \SetKw{brk}{break}
  \SetKw{chse}{choose}
  \KwIn{$\I = \langle \D,\C, \P, \mathit{ROLs} \rangle$ an \textbf{\smpc} instance}
  \KwOut{Enumerate all stable models of a $\I$}
  CNF $\mbox{} \leftarrow \mbox{}$ SAT-E(\I)\;
  \While{\true}{
    (sat?,$\pi$) $\mbox{}\leftarrow\mbox{}$  SatSolve(CNF)\;
    \tcc{SatSolve returns the status (sat or unsat) and a satisfying model $\pi$ if sat}
    \If{sat?}{
      $\match \leftarrow \mbox{}$ stable matching corresponding to $\pi$\;
      $c \leftarrow \{ \lnot m_d[p]\, |\, \match(d) = p\}$\;
      CNF $\mbox{}\leftarrow \mbox{}$ CNF $\mbox{} \cup \{c\}$\; 
      \textbf{enumerate}($\match$)\;
    }
    \Else{
      \Return{} \tcp*{All stable matchings enumerated.}
    }
  }
  \caption{
    Enumerate all stable models of an inputted \smpc instance}
  \label{algo:all_matches}
\end{algorithm}

Algorithm~\ref{algo:all_matches} is our algorithm for enumerating all
stable models. It uses a sequence of calls to a SAT solver to do this.
It first constructs the SAT-E encoding of the \smpc instance then
enters a loop. Each time through the loop a new stable model is found
by the SAT solver, and a \emph{blocking clause} $c$ is added to the
SAT-E encoding. This clause $c$ is a disjunction that says that no
future solution is allowed to return the same stable model (one of the
mappings $\mu(d) = p$ of the corresponding stable model must be
different, i.e., one of the variables $m_d[p]$ made true by $\mu$ must
be false in every future matching).

\begin{algorithm}[t]
  \dontprintsemicolon
  \SetCommentSty{textsl}
  \SetNoLine
  \SetKwComment{Comment}{\textbf{/* }}{\textbf{*/}}
  \SetKw{brk}{break}
  \SetKw{chse}{choose}
  \KwIn{$\I = \langle \D,\C, \P, \mathit{ROLs} \rangle$ an
    \textbf{\smpc} instance
  and $\match$ a stable matching for $\I$}
  \KwOut{An $\RPopt$ matching for $\I$ that dominates $\match$}
  \While{\true}{
    CNF $\mbox{} \leftarrow \mbox{}$ SAT-E(\I)\;
    $c \leftarrow \{ \lnot m_d[p]\, |\, \match(d) = p\}$\;
    CNF $\mbox{}\leftarrow \mbox{}$ CNF $\mbox{} \cup \{c\}$\; 
    \For{\,$d \in \S$}{
      $c_d \leftarrow \{ m_d[p]\, |\,  p \pgeq{d} \match(d)\}$\;
      CNF $\mbox{}\leftarrow \mbox{}$ CNF $\mbox{} \cup \{c_d\}$\; 
    }
    \For{$c\in C$}{
      $c_c \leftarrow \{ m_c[i]\, |\, i = \rank{c}(\match(c))\}$\;
      CNF $\mbox{}\leftarrow \mbox{}$ CNF $\mbox{} \cup \{c_d\}$\;
    }
    (sat?,$\pi$) $\mbox{}\leftarrow\mbox{}$  SatSolve(CNF)\;
    \If{sat?}{
      $\match \leftarrow \mbox{}$ stable matching corresponding to $\pi$\;
    }
    \Else{
      \Return{\match} \tcp*{Return last match found.}
    }
  }
  \caption{
    Given an \smpc instance and a stable matching $\match$ find a
    dominating $\RPopt$ matching.}
  \label{algo:rpopt_match}
\end{algorithm}

Algorithm~\ref{algo:rpopt_match} is our algorithm for finding a
$\RPopt$ matching that dominates an inputted stable matching $\match$.
The algorithm takes an $\smpc$ instances as input and constructs the
SAT-E encoding for that instance. It then blocks the match $\match$
from being a satisfying solution (using the same kind of clause as
Algo.~\ref{algo:all_matches}) and also forces the next match found to
be $\mgeq$ the current match. It accomplishes this by adding a clause
for each single doctor $d$ that says that $d$ must be matched to a
program it ranks at least as high as $\match(d)$, and for every couple
$c$ a (unit) clause that says that $c$ must be matched to a pair of
programs it ranks at least as highly as $\match(c)$. This causes the
new match to be $\mgeq$ to the current match, and since the new match
cannot be equal it must be $\mgt$ the current match. That is, the new
match must dominate the current match. If no dominating match can be
found, the current match is $\RPopt$ and we return it.

\section{Empirical Results}
\label{sec:empirical}
In this section we report on various experiments we performed. We used
the state-of-the art SAT solver LingeLing \cite{biere2013lingeling} to
solve our encoding SAT-E; Algorithm~\ref{algo:all_matches} to find all
stable matches; and our implementation of the two DA algorithms RP99
and KPR, described in Section~\ref{sec:da}, when testing the
effectiveness of DA algorithms.

\subsection{Statistical models}
We experiment with randomly generated \smpc instances. In our
experiments we confine our attention to instances in which all program
quotas are 1. So these are one-to-one matching problems with couples.

We generate ordered lists from sets using a common random sampling
procedure. To obtain an ordered list $L$ of size $k$ from a set $S$,
given that we already generated the first $i$ items, we draw elements
from $S$ independently and uniformly at random until we find an
element $e$ that does not already appear in $L$. Once we have selected
such an $e$ it becomes the $i+1$'th item of $L$. This process stops
when $L$ has $k$ items.

We use market sizes of size $n$ where $n$ ranges from 200 to 20,000,
and varying percentages $x$ of couples. For each $n$ we include $n$
doctors in $\D$ and $n$ programs in $\P$. Among the doctors we mark
$1-x$ percent as being singles and the remaining $x$ percent are
paired into couples. (Thus we have $x*n/2$ couples and $n-x*n$
singles). For each single we randomly generate an ordered ROL of
length 5 from $\P$ (using the procedure described above), each couple
has an ROL of length 15 randomly generated from
$\Pplus \times \Pplus - \{(\nomatch, \nomatch)\}$, and each program
has an ROL that includes all doctors that ranked it (including one
member of a couple) and is randomly generated from that same set. For
each setting of the parameters, we drew 50 problem instances.

\subsection{Experiments}
\subsubsection{Existence of a Stable Matching in \smpc}
Our first experiment involves running a SAT solver on each \smpc
instance. This allows us to determine how often these instances had a
stable match. Figure \ref{fig:immorlica_pct_satisfiable} shows the
trend. We see that as the fraction of residents who are in couples
increases, the fraction of satisfiable instances drops fairly rapidly.
Interestingly, it appears as though this effect is independent of the
size of the market; for all market sizes drawn, the fraction of
satisfiable instances is fairly consistent, and follows the same
trend. These results agree with the predictions of Ashlagi et al. who
showed that in a random model similar but not identical to ours
whenever the number of couples grows linearly with market size, the
probability of no stable matching existing is a constant. This agrees
with our data that market size has little effect. There model was not,
however, able to predict what that probability might be, nor how it
might change as the percentage of couples increases. 

It can also be noted that the results presented here (20,000 couples
with an average single's ROL of size 5) are of a size comparable to
the NRMP, where there are 34,355 residents with an average single's
ROL of size 11 \cite{nrmp_data:2013}. On NRMP Roth and Peranson
\scite{roth99} remark that historically no instance was found not to
have a stable matching. This might arise from (a) the percentage of
couples being very small for NRMP, or (b) extra structure in this real
problem not present in our random problems. If the cause is (a) there
might be practical concerns for the ability of clearinghouse
mechanisms to find stable matchings if the percentage of couples
rises.

\begin{figure}
\center
\includegraphics[width=1.0\columnwidth]{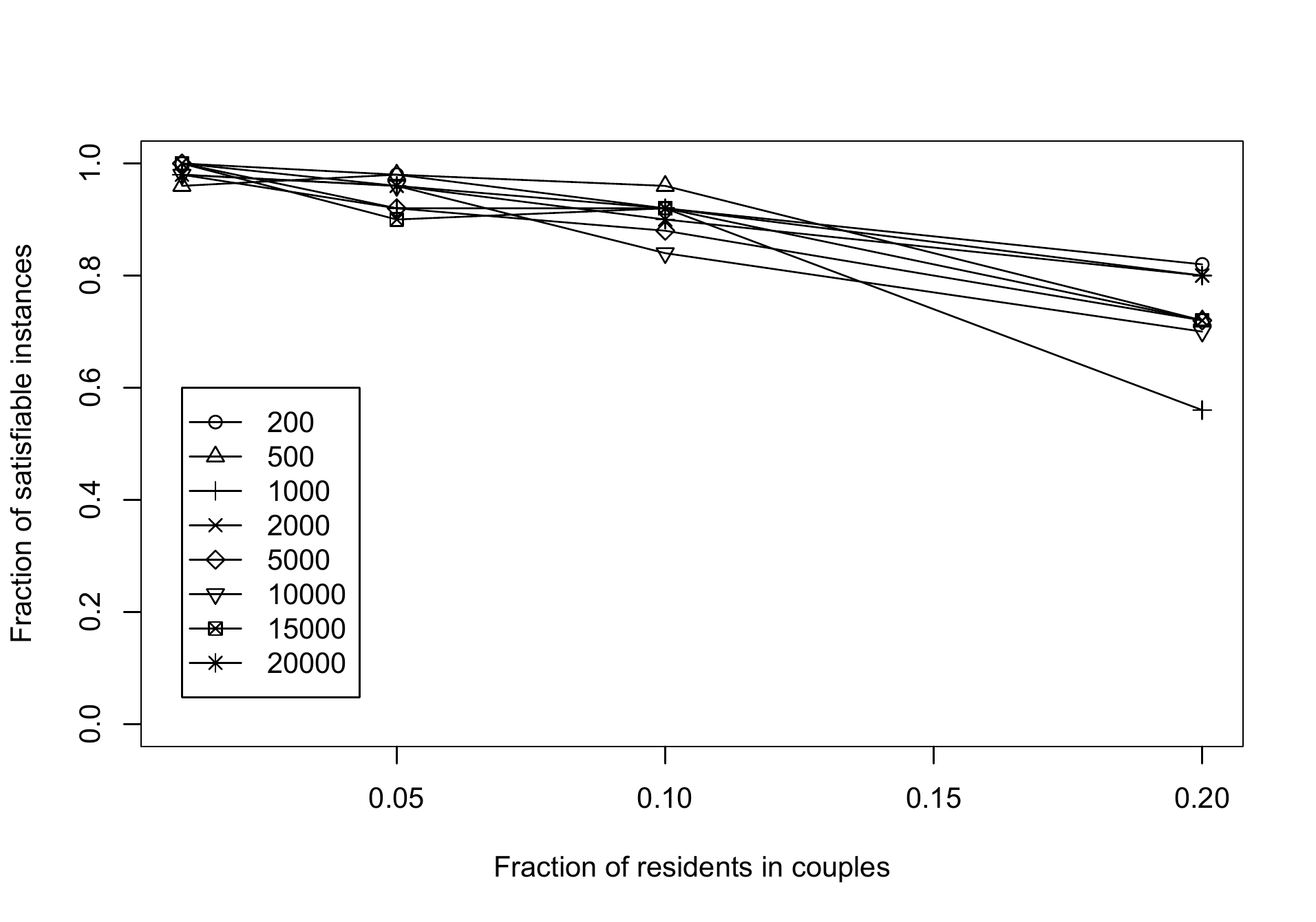}
\caption{Fraction of satisfiable instances drawn, varying the fraction of couples in the match, and the size of the market.} 
\label{fig:immorlica_pct_satisfiable}
\end{figure}

\subsubsection{Existence of a Unique Stable Matching in \smpc}
Figure \ref{fig:immorlica_pct_unique} shows the fraction of instances
that had a unique stable matchings as the size of the market
increases, and as the percentage of the couples in the market
increases. Importantly, this figure only counts the fraction of
instances with a unique matching among those instances that have at
least one matching. This data was generated by running
Algo.~\ref{algo:all_matches}.

As shown in Figure \ref{fig:immorlica_pct_satisfiable}, the fraction
of instances where a stable matching exists changes as the percentage
of couples in the market changes. So to remove that underlying trend,
we only look instances where a stable matching is known to exist.

Note that, while these curves generally mimic the pattern we'd expect
to see (the fraction of unique matchings tends to increase as the
market size gets larger), the curve appears to level out and reach an
asymptote much earlier than expected. In our experiments, no setting
reached more than 84\% unique matchings.

Immorlica and Mahdian proved that, for \smp (where at least one stable
matching is always true) with fixed length ROLs (and everyone drawing
from some given distribution), the chance of drawing a problem with
only one stable matching goes to 1 as the size of the market goes
to infinity \scite{immorlica2005marriage}.

Our trends for \smpc, however, appear somewhat asymptotic, especially
when the percentage of couples is low. It thus appears that, at least
for the size of problems relative to the size of ROL investigated
here, that the Immorlica and Mahdian results do not apply. This could
be because their results give very large bounds that require the size
of the problems to be much larger than drawn here. So our results
might not be on sufficiently large markets to capture the phenomenon
they describe. Alternately, this could be because their results do not
hold for \smpc.

As an aside, when comparing the fraction of unique satisfiable
instances out of all problem instances drawn, the numbers look even
worse; with a market size of 20,000 and 20\% couples, the percentage
of problems drawn with only one stable matching drops to 50\%.


\begin{figure}
\center
\includegraphics[width=1.0\columnwidth]{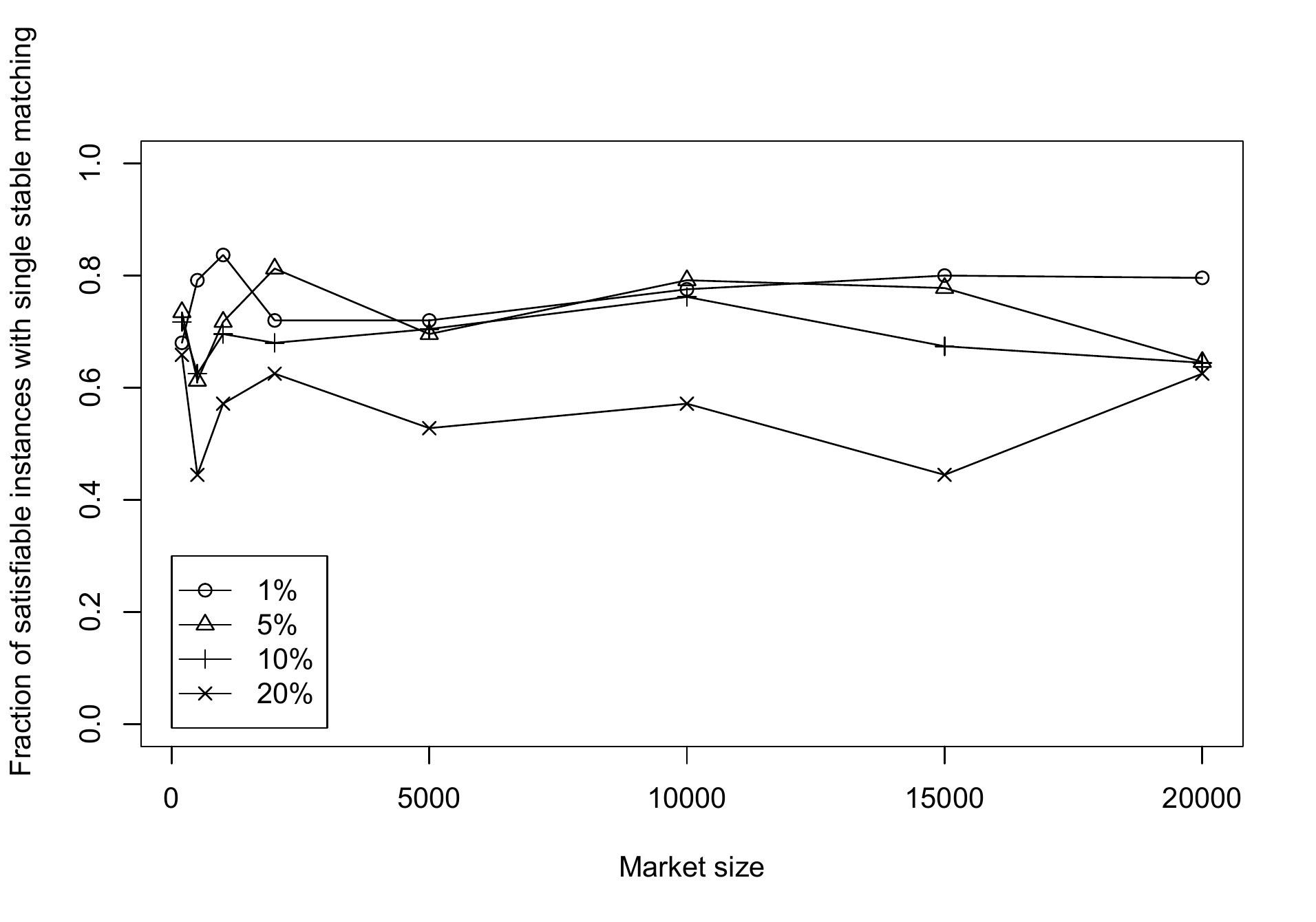}
\caption{Fraction of satisfiable instances with single stable matching, varying size of the market and percentage of couples in the market.}
\label{fig:immorlica_pct_unique}
\end{figure}


\subsubsection{Properties of the Set of Matchings for \smpc}
We next investigate the properties of the set of matchings of each
\smpc instance as generated by Algorithm~\ref{algo:all_matches}. We
have shown that, at least under the settings explored in this paper,
unique stable matches are not as common as we may have hypothesized.
We also have proven that, even if a unique matching exists (and hence
it is an $\Ropt$ matching), that is no guarantee of strategy-proofness
on either the residents' or programs' side for \smpc.

Nevertheless, we have also shown that if residents are only allowed to
truncate their preferences (which is arguably the easiest misreporting
strategy), then they have no incentive to misreport if an $\Ropt$
matching exists. Since an $\Ropt$ matching can exist even when there
more than one stable matching, this can be a useful situation that
covers more cases than having only one stable matching.

Figure \ref{fig:immorlica_avg_size} shows our results about this
question. The boxes in this plot represents the average number of
stable matches for different problem parameters; the whisker
represents one standard deviation. For the problems we investigated,
the average number of stable matchings was quite low (1.30), and did
not seem to be affected much by either the size of the market, or the
percentage of couples involved in the market. Note that exceedingly
few problem instances have more than 3 stable matches, no matter how
the parameters are set. Also, in each setting of the parameters,
problem instances with one and two stable matchings are very common.
Note that, unlike the previous analysis presented in Figure
\ref{fig:immorlica_pct_satisfiable}, we are looking at \emph{all}
problem instances, not just the satisfiable ones.

In the problems we investigated, 10.63\% had no stable matching, 
61.91\% had a unique stable matching,
21.89\% had two stable matchings,
0.31\% had three stable matchings,
4.69\% had four stable matchings,
0.13\% had six stable matchings, and
0.44\% had eight stable matchings.

\begin{figure}
\center
\includegraphics[width=1.0\columnwidth]{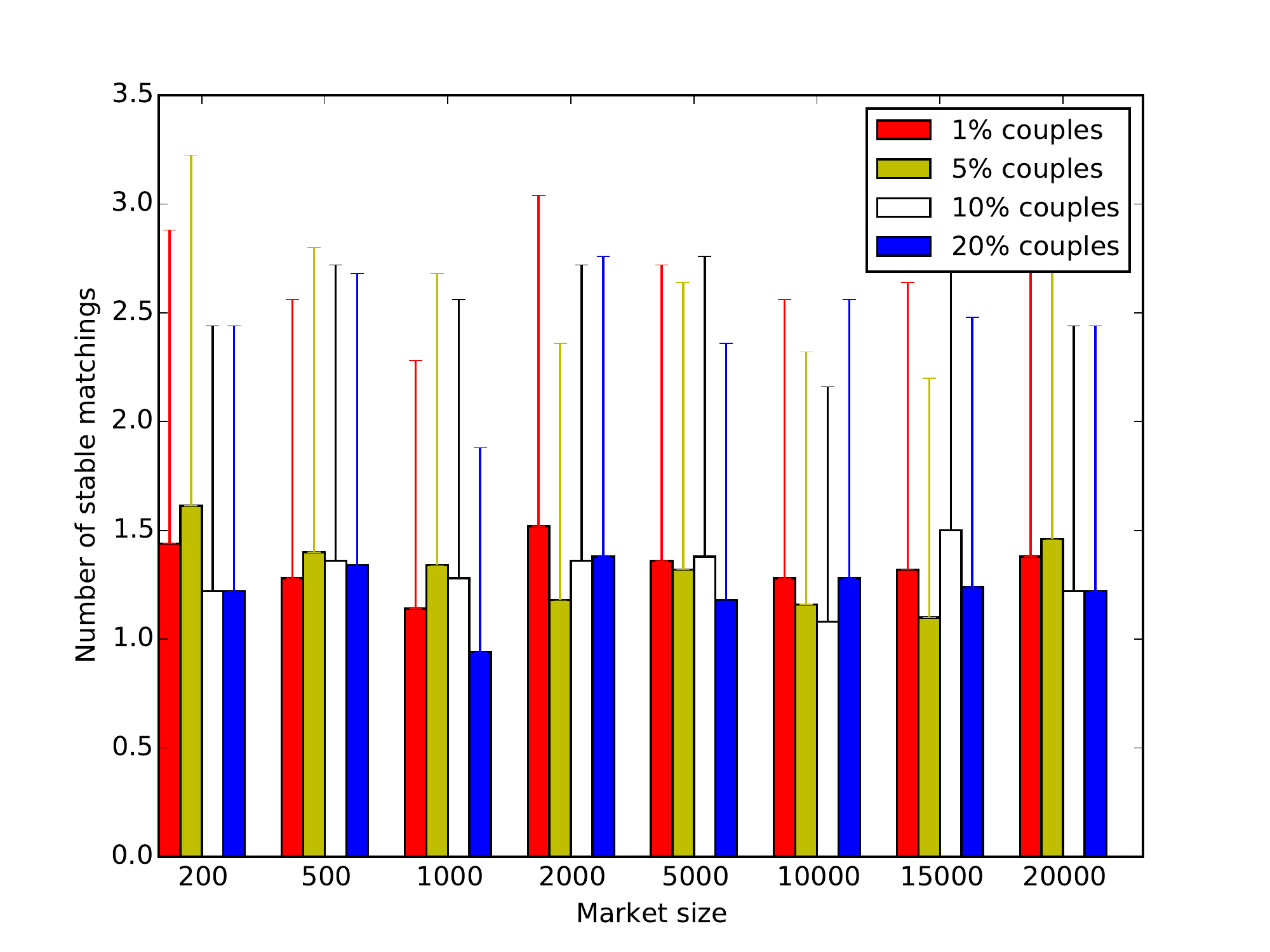}
\caption{Number of stable matchings versus market size and percentage of couples present in the market. Box represents average number of stable matchings; whisker represents one standard deviation}
\label{fig:immorlica_avg_size}
\end{figure}

When a problem instance has more than one stable matching, we are
particularly interested in how many $\RPopt$ matchings exist. If only
one $\RPopt$ matching exists this matching is also an $\Ropt$
matching, and our theoretical result applies: residents have no
incentive to misreport their preferences by truncation. 78.5\% of all
problem instances drawn had an $\Ropt$ matching, which is many more
instances than simply had a single stable matching (61.91\%).
Additionally, if we again look at the Immorlica and Mahdian results,
but with respect to fraction of satisfiable instance that had an
$\Ropt$ matching instead of the more restrictive condition of having
only one stable matchings, we see a trend much closer to what we might
expect. See Figure \ref{fig:immorlica_pct_rposm}. Furthermore, many
instances with a low percentage of couples almost always had an
$\Ropt$ matching; when the market had 1\% couples, 98.5\% of all
satisfiable instances had an $\Ropt$ matching; with 5\%, 92.7\% had an
$\Ropt$ matching. Importantly, this means that for many of the problem
instances that we drew, residents have no incentive to misreport their
preferences via truncation. However, further investigation is required
to ascertain the relationship between the percentage of couples in a
matching and the fraction of instances that have an $\Ropt$ matching.

\begin{figure}
\center
\includegraphics[width=1.0\columnwidth]{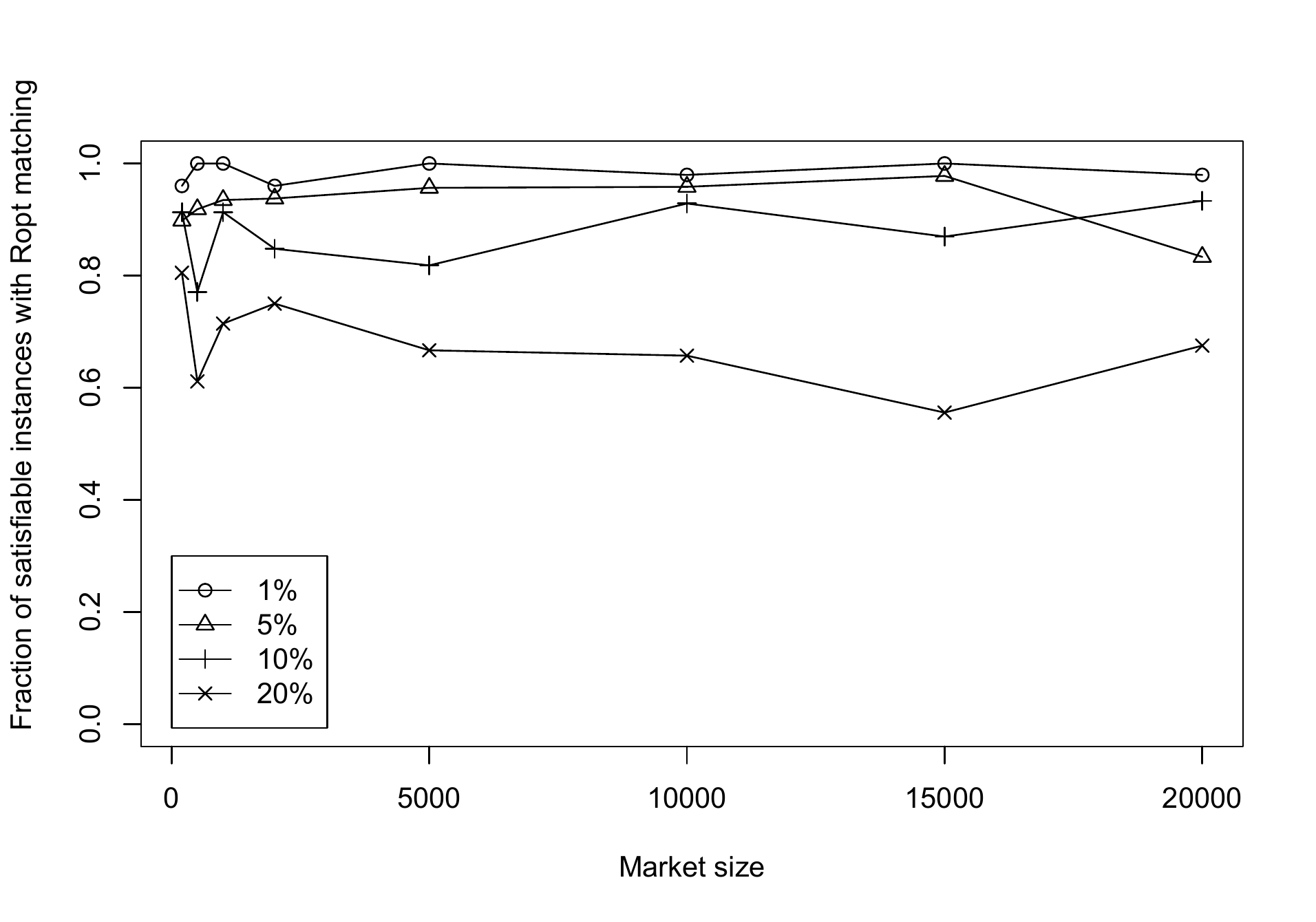}
\caption{Fraction of satisfiable instances with an $\Ropt$ matching, varying size of the market and percentage of couples in the market.}
\label{fig:immorlica_pct_rposm}
\end{figure}

Finally, looked at those instances with more than one stable matching
to see how many, on average, $\RPopt$ matchings they contain. We found
that in some cases $\RPopt$ matchings are quite common. When there are
two stable matchings for a given problem instance, there were an
average of 1.36 $\RPopt$ matchings per problem. Instances with 3
stable matchings always had 3 $\RPopt$ matchings per problem.
Instances with 4 stable matchings had an average of 1.69 $\RPopt$
matchings per problem, and instances with 8 stable matchings had an
average of 4 $\RPopt$ per problem.

\subsection{Structure of \smpc Problem Instances}
In this section we examine some finer grained structure among the set
of stable matchings. We examined the stable matchings and a basis for
the set of Pareto improving moves. A Pareto improving move is the set
of resident transfers needed to transform matching $\mu_2$ to matching
$\mu_1$ when $\mu_1\mgt \mu_2$ (i.e., the transfers improve the
match). Improving moves that are the result of combining other
improving are not counted as part of the basis. 

We show this structure as a directed graph, where nodes are stable
matchings for a given instance, and an edge goes from node $a$ to $b$
if there is a Pareto improving move transforming $a$ to $b$. (Note
this implies that $b\mgt a$. Note this graph is not guaranteed to be
connected, and in practice, many of these graphs have multiple
components. By definition, each component will contain at least one
$\RPopt$ matching. While it appears to be exceedingly rare for any
component to contain more than one component (at least in the problems
we generated), we did find one component with two $\RPopt$ matchings.

For an illustration of what these components and the possible Pareto
improving moves may look like, see Figure \ref{fig:immorlica_8_fig}.
For the problem instance chosen, there were 8 stable matchings, 2 of
which were $\RPopt$ matchings. Each node in the graph contains a tuple
of the singles' and couples' average rank; each edge on the graph is
the maximum improvement for any single and any couple from the
previous match to the new Pareto improved match.

Also note that in this problem instance, on average, one of these
$\RPopt$ matchings is better for couples (by 0.03 positions on
average) and one of these $\RPopt$ matchings is better for singles (by
0.0002 positions on average). Interestingly, there are very few unique
Pareto optimal moves in this graph; there are only two. This tends to
hold in general; when there are two stable matchings, there are an
average of 0.64 unique Pareto moves per problem. With 4 stable
matchings, there are an average of 1.48 Pareto moves.

Note that, for the single that has the greatest incentive to go from
one matching to another, they have quite a high incentive---almost
their entire ROL. (Remember, singles give a ROL of size 5.) Note that
Figure \ref{fig:immorlica_8_fig} shows a case where only single
residents improve; this is not true in the general case. Couples,
likewise, can greatly benefit from moving from a non-Pareto optimal
matching to a $\RPopt$ matching. Additionally, some singles and
couples will strongly prefer one $\RPopt$ matching over another; in
this example,
the maximum incentive for any single to switch from the $\RPopt$
matching on the left to the one on the right is moving up their ROL 3
positions. In this example, there is no incentive for couples to
switch from the $\RPopt$ matching on the left to the one on the right.
For switching from the $\RPopt$ matching on the right to the one on
the left, the maximum incentive for any single to move is one position
in their ROL, and 3 positions for any couple. In this instance, both
of the DA algorithms KPR and RP99 found the same $\RPopt$ matching,
denoted by the square node in the graph. While KPR and RP99 are not
guaranteed to find an $\RPopt$ matching, they did in this instance.

\begin{figure}
\center
\includegraphics[width=1.0\columnwidth]{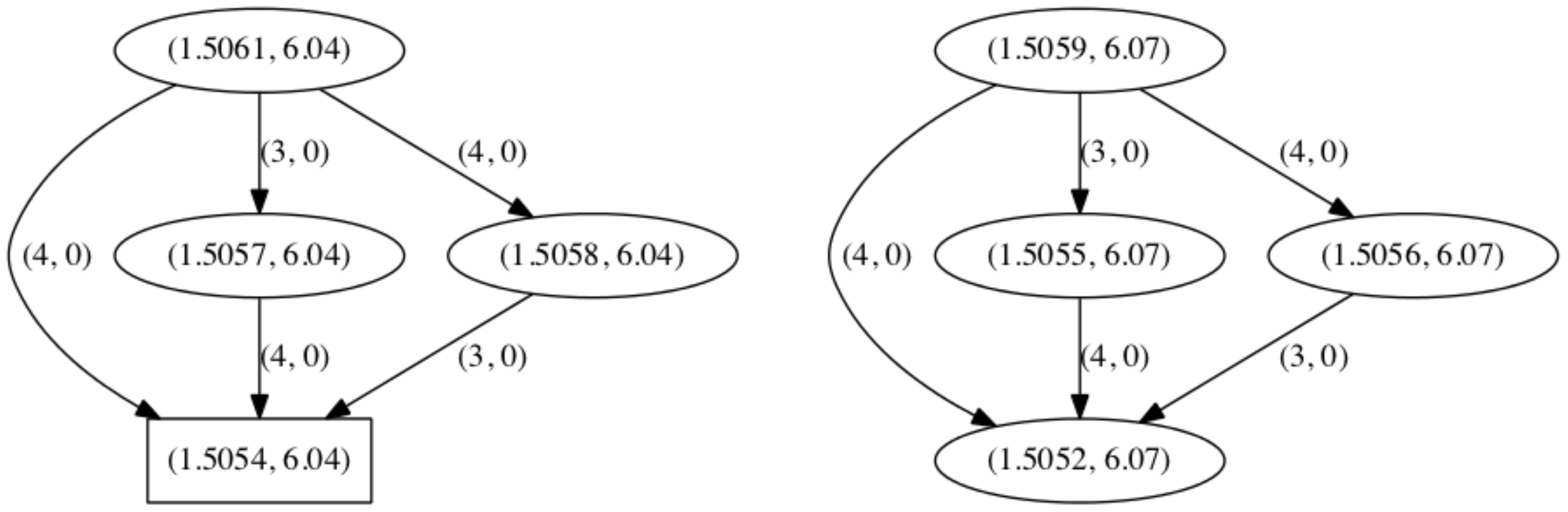}
\caption{Pareto Improvement graph for one problem instance, with 8
  stable matchings, and 2 $\RPopt$ matchings. Square denotes the
  matching KPR and RP99 found.}
\label{fig:immorlica_8_fig}
\end{figure}

While in this example, there is little incentive to switch from one
$\RPopt$ matching to another,
there were instances in the data where at least one single could move
up their \emph{entire} preference list by switching from one $\RPopt$
matching to another, and at least one couple could move up their
\emph{entire} preference list by switching from one $\RPopt$ matching
to another.

\subsection{The Performance of DA Algorithms for \smpc}
We are also interested in how the standard DA-style algorithms perform
with respect to finding $\RPopt$ matchings. As Roth and Peranson
mention that resident optimality is an important design aspect for
stable matching algorithms, we investigate how frequently the DA
algorithms find a $\RPopt$ matching. We do this by using
Algo~\ref{algo:all_matches} to enumerate all possible stable matchings
for each problem instance. We then identify all $\RPopt$ matchings,
and then see if the DA algorithms find a $\RPopt$ matching for that
instance. Table \ref{tab:da_perf} shows these results. We did not find
much of a relationship between the size of the problem and how
frequently the DA algorithms found an $\RPopt$ matching, so we combine
all problem sizes for this analysis. Note that, when analyzing the \%
of times the DA algorithm found an $\RPopt$ matching, we only look at
problems where (a) more than one stable matching exists (if only one
stable matching exists, that solution is automatically an $\RPopt$
matching) and (b) the DA algorithm found a matching (the DA algorithms
have different failure rates and we wanted to directly compare only
when they succeed.) Table \ref{tab:da_perf} also shows the general
failure rate for the DA algorithms (i.e., when at least one stable
matching exists, but the algorithm fails to find one).

\begin{table}
\label{tab:da_perf}
\begin{tabular}{|r||r|r|r|r|}
\hline
 & \multicolumn{2}{|c|}{\% $\RPopt$ found} & \multicolumn{2}{|c|}{\% Total Timeouts}\\
\cline{2-5}
\% Couples & KPR &RP99  & KPR & RP99 \\
\hline
\hline
1 & 100.00 & 100.00 & 0.00 & 0.00 \\
\hline
5 & 94.54 & 96.16 & 0.25 & 0.75 \\ 
\hline
10 & 88.66 & 89.58 & 0.25 & 6.25 \\ 
\hline
20 & 91.42 & 90.24 & 5.57 & 36.26 \\ 
\hline
\end{tabular}
\caption{DA algorithms performance versus percentage of couples in a match. For \% $\RPopt$ found, only problems with more than one stable matching that did not timeout were considered.  For \% Total timeouts, all problems with at least one stable matching were considered.}
\end{table}

Critically, note that even though the DA algorithms were designed to
be as resident favoring as possible, they do not find an $\RPopt$
matching in a large percentage of problems with multiple stable
matchings. When the percentage of couples is low, the DA algorithms'
performance is either perfect, or fairly good (with about a 5\%
failure rate). For problems with roughly the same percentage of
couples as the NRMP (roughly 6\%) the failure rate of KPR and RP99 to
find a $\RPopt$ matching is fairly low, at 5\% (though KPR's failure
rate is slightly higher.) However, with more couples in the match, the
failure rate jumps to roughly 10\%. Again, in all of these settings,
SAT-E with the additional constraints was able to enumerate \emph{all}
stable matchings, including all $\RPopt$ matchings.

It appears as though RP99 finds an $\RPopt$ matching
with a slightly higher rate than KPR when it finishes. These two
algorithms deal with couples fairly differently (as RP99 has couples
propose individually, and couples propose in-batch in KPR), which
could account for the difference in performance. However, KPR
outperforms RP99 by a large margin with respect to the percentage of
problem instances solved. KPR's failure rate is less than 0.25\% for
all problems with 10\% or less couples, and only 5.75\% when the match
has 20\% couples. RP99's failure rate grows quickly, reaching 36.25\%
by the time couples consist of 20\% of the market. Therefore, it seems
as though there is a slight benefit to using RP99 to find $\RPopt$
matchings, but its large failure rate may negate this benefit.





\section{Discussion and Future Work}
In this paper, we have provided new theoretical results on
strategy-proofness for \smpc, and new extensions to an existing SAT
encoding for \smpc that allows us to both enumerate every stable
matching for a given problem instance, and guarantee that we've found
a $\RPopt$ matching if we don't have the resources to enumerate all
instances. We then used these techniques to explore properties of \smpc
instances.

We used these tools to empirically investigate previous theoretical
results. Interesting, and perhaps unexpectedly, the Immorlica and
Mahdian results (given a constant list size, the probability of a
unique stable matching goes to one as the market size increases) did
not seem to hold. However, we empirically found evidence for a related
phenomenon for \smpc, leading us to the following conjecture: for
problem instances drawn from the same distribution with a fixed ROL
size of $k$, the probability of a resident optimal stable matching
goes to one as the size of the market goes to infinity.

We also empirically evaluated RP99 and KPR, two DA-style algorithms
for \smpc that do not have any theoretical guarantees. While they
frequently perform well on the problem instances we drew, we found
many problem instances where they either did not find a stable
matching at all even though one existed, or if they did find one, they
did not find a $\RPopt$ matching. In practice, it would be a simple to
use our algorithm~\ref{algo:rpopt_match} to improve the matching found
by these algorithms.

By enumerating all possible stable matchings (and thus all resident
Pareto optimal matchings for
a given matching), we can now pick stable matchings based on certain
properties. Enumerating the set of all stable matchings allows us to
pick the matching with the best average rank for all singles and
couples (a sort of resident social-welfare optimal matching), or the
stable matching that places the worst-off individual as high in their
preference rankings as possible. While we did not empirically evaluate
this in this paper, being able to enumerate all stable matchings
allows for a rich investigation into this problem.

Our overall conclusion is that encoding to SAT is a powerful tool for
further empirical analysis of stable matching problems. The research
question that arises is how best to use this tool and how this tool
can be adapted to obtain insight into other important open questions
about NP-complete stable matching problems. The current paper has
shown that this tool has already exposed some interesting theoretical
questions. For example, what can be proved about the frequency in
which resident-optimal matchings appear in randomly drawn \smpc
instances as the market size increases; is there further structure to
the set of $\RPopt$ solutions; can the existence of a resident-optimal
matching or conditions on the set of $\RPopt$ solutions be used to
prove further strategy-proof results for \smpc.

In addition to theoretical results, there is considerable potential
for practical computation of stable matchings. For example, using both
DA algorithms, and algorithm~\ref{algo:rpopt_match} to improve the
resulting matching might have some very useful applications. It might
also be the case that solving SAT-E with additional constraints might
be applicable to some currently unsolved practical matching problems. 

\vspace{1mm}
\noindent \textbf{Acknowledgements.}
This work is supported by the Natural Sciences and Engineering
Research Council of Canada.  
Perrault and Drummond are additionally supported by an Ontario Graduate Scholarship.

\bibliographystyle{named}
\bibliography{sat_smpc_optimality}

\newpage
\appendix
\section{SAT Encoding for the Stable Matching Problem with Couples}
\label{sec:sat-encoding}
In this appendix we provide an encoding of \smpc into SAT. \smpc was
shown to be an NP-Complete problem by Ronn \cite{ronn1990np}. Hence, it can be
encoded and solved as a satisfiability problem (SAT). This is an attractive
approach for solving \smpc due to the recent tremendous advances in SAT solvers.

The encoding of \smpc we give here forms the basis of Drummond et al.\ \shortcite{dpb_review}. Nevertheless, since we extend this
SAT encoding in order to find Pareto optimal matchings and to find all stable
matchings, we provide its details in this appendix (which will only be available
on-line).

First we assume that the doctor (singles and couples) \rols have been
preprocessed so as to remove from them any program that does not find that
doctor acceptable. For $d\in \S$ we remove $p$ from $\rol{d}$ if
$d\notin \rol{p}$. For couple $(d_1,d_2)\in \C$ we remove $(p_1, p_2)$ from
$\rol{(d_1,d_2)}$ if $d_1 \notin \rol{p_1}$ or $d_2 \notin \rol{p_2}$. This
ensures that in any matching found by solving the SAT encoding no doctor $d$ can be
matched into a program that does not rank $d$. 

A SAT encoding consists of a set of Boolean variables and a collection of
clauses over these variables. We first describe the set of variables and their
intended meaning, and then we present the clauses which enforce the conditions
required for a stable matching.\\[10pt]

\subsection{Variables}
We utilize three sets of Boolean variables.

\noindent
\textbf{1. Doctor Matching Variables:}
\begin{align*}
\{\matchDvar{d}{p}\,|\, d\in D \land p\in \rol{d}\}
\end{align*}
$\matchDvar{d}{p}$ is true iff $d$ is matched into program $p$. Note
that $\matchDvar{d}{\nomatch}$ is true if $d$ is unmatched.

\noindent
\textbf{2. Couple Matching Variables:}
\begin{align*}
\{\matchCvar{c}{i}\,|\, c\in C \land
(0 \leq i < |\rol{c}|)\}
\end{align*}
$\matchCvar{c}{i}$ is true iff couple $c$ is matched into a program
pair $(p,p')$ that it ranks between $0$ and $i$, i.e., $0 \leq
\rank{c}((p,p')) \leq i$. (Lower ranks are more preferred).

\noindent
\textbf{3. Program Matching Variables:}
\begin{align*}
\{\matchPvar{p}{i}{s}\, |\, p\in \P &\land (0 \leq i \leq |\rol{p}|-2) \\ 
&\land (0 \leq s \leq \min(i+1, \quota{p}+1))\}
\end{align*}
$\matchPvar{p}{i}{s}$ is true iff $s$ of the doctors in $\rol{p}[0]$
to $\rol{p}[i]$ have been matched into $p$. Note that $i$ ranges up to
$|\rol{p}|-2$ which is the index the last non-$\nomatch$ doctors on
$\rol{p}$ ($\rol{p}$ is terminated by $\nomatch$).

The program matching variables are the main innovation of our encoding. These
variables are especially designed to allow us express more efficiently the
constraints a stable matching must satisfy.

\subsection{Clauses}
Now we give the \textbf{clauses} of the encoding. Rather than give the
more lengthy clauses directly, we often give higher level constraints
whose CNF encoding is straightforward.

\paragraph{Unique Match} A doctor must be matched into exactly one
program (possibly the \nomatch program). For all $d\in D$ 
\begin{description}
  \item[\textbf{1a}] $\mbox{\textbf{at-most-one}}(\{\matchDvar{d}{p} | p\in
  \ranked{d}\})$
  \item[\textbf{1b}] $\bigvee_{p \in \ranked{d}} \matchDvar{d}{p}$
\end{description}

The \textbf{at-most-one} constraint simply states that at most one of a set of
Boolean variables is allowed to be true. It can be converted to CNF in a number
of different ways. \textbf{1b} ensures that some match (possibly to the \nomatch
program) is made.

\paragraph{Couple Match} The $\matchCvar{c}{*}$ variables must have
their intended meaning. For all couples $c \in \C$, for all $k$ such
that $1\leq k \leq |\rol{c}|$, letting $c=(d_1, d_2)$ and $(p_1[i],
p_2[i]) = \rol{c}[i]$, 
\begin{description}
  \item[\textbf{2a}]
  $ \matchCvar{c}{0} \equiv \matchDvar{d_1}{p_1[0]} \land
  \matchDvar{d_2}{p_2[0]}$
  \item[\textbf{2b}]
  $ \matchCvar{c}{k} \equiv (\matchDvar{d_1}{p_1[k]} \land
  \matchDvar{d_2}{p_2[k]}) \lor \matchCvar{c}{k-1}$
  \item[\textbf{2c}]
  $ \matchCvar{c}{|\rol{c}|}$
\end{description}
The final condition ensures that $c$ is matched to some program pair
on its \rol{} (possibly $\nomatch$), and the 
\textbf{at-most-one} constraint for $d_1$ and $d_2$ 
ensures that $c$ is uniquely matched.

\paragraph{Program Match} The $\matchPvar{p}{*}{*}$ variables must have
their intended meaning. For all programs $p\in \P$, for all $i$ such
that $1 \leq i \leq |\rol{p}|-2$, and for all $s$ such that $0 \leq s
\leq \min(i+1, \quota{p}+1)$, letting $d_i = \rol{p}[i]$ (the $i$-th
doctor on $p$'s \rol{}),
\begin{description}
  \item[\textbf{3a}]
  $\bigl(\matchPvar{p}{0}{0} \equiv  \lnot
  \matchDvar{d_0}{p}\bigr) \land \bigl(\matchPvar{p}{0}{1}
  \equiv  \matchDvar{d_0}{p}\bigr)$
  \item[\textbf{3b}]
    $\matchPvar{p}{i}{s} \equiv
    (\matchPvar{p}{i{-}1}{s}{\land}\lnot \matchDvar{d_i}{p})
    \lor (\matchPvar{p}{i{-}1}{s{-}1}{\land}\matchDvar{d_i}{p})$
  \item[\textbf{3c}]
  $\lnot \matchPvar{p}{i}{\quota{p}+1}$
\end{description}
The last condition, captured by a set of unit clauses, ensures that 
$p$'s quota is not exceeded at any stage. Falsifying these variables
along with the other clauses ensures that no more matches can be made
into $p$ once $p$ hits its quota. Note that $i$ only indexes up to the
last non-\nomatch doctor on $\rol{p}$ since $\nomatch$ does not use up
any program capacity.


\paragraph{Stability for Singles} For each single doctor, $d\in S$ and
for each $p\in \rol{d}$ 
\begin{description}
  \item[\textbf{4}] $
  \bigl(\bigvee_{p' \pgeq{d} p} \matchDvar{d}{p'}\bigr) \lor
  \matchPvar{p}{\rank{p}(d)-1}{\quota{p}}$
\end{description}

This clause says that if $d$ has not been matched into a program
preferred to or equal to $p$, then it must be the case that $p$ will
not accept $d$. Note that $\matchPvar{p}{\rank{p}(d)-1}{\quota{p}}$
means that $p$ has been filled to capacity with doctors coming before
$d$ on its \rol.

\paragraph{Stability for Couples (A)} For each couple $c = (d_1,d_2) \in
\C$ and for each $(p_1,p_2) \in \rol{c}$ \textbf{with}
$\mathbf{p_1\neq p_2}$
\begin{small}
\begin{description}
  \item[\textbf{5a1}] $  
    \matchDvar{d_1}{p_1} \land \lnot \matchCvar{c}{\rank{c}((p_1,p_2))} \\
    \implies \matchPvar{p_2}{\rank{p_2}(d_2)-1}{\quota{p_2}}$
  \item[\textbf{5a2}] $
    \matchDvar{d_2}{p_2} \land \lnot \matchCvar{c}{\rank{c}((p_1,p_2))} \\
    \implies  \matchPvar{p_1}{\rank{p_1}(d_1)-1}{\quota{p_1}}$
  \item[\textbf{5b}] $ 
  \begin{array}[t]{l}
    \lnot \matchDvar{d_1}{p_1} \land \lnot \matchDvar{d_2}{p_2} \land
    \lnot \matchCvar{c}{\rank{c}((p_1,p_2))}
    \mbox{} \, \implies \\
      \matchPvar{p_1}{\rank{p_1}(d_1)-1}{\quota{p_1}} 
      \lor \matchPvar{p_2}{\rank{p_2}(d_2)-1}{\quota{p_2}}
  \end{array}$
\end{description}
\end{small}

Clause \textbf{5a1} says that when $d_1$ is already matched to $p_1$
but $c$ has not been matched into $(p_1,p_2)$ or into a more preferred
program pair, then it must be the case that $p_2$ will not accept
$d_2$. \textbf{5a2} is analogous.

Clause \textbf{5b} says that if neither $d_1$ nor $d_2$ is matched
into $p_1$ or $p_2$ and $c$ has not been matched into $(p_1,p_2)$ or
into a more preferred program pair, then either $p_1$ will not accept
$d_1$ or $p_2$ will not accept $d_2$. 

\paragraph{Stability for Couples (B)} For each couple $c=(d_1,d_2) \in
\C$ and for each $(p,p)\in \rol{c}$ we have one of constraint
\textbf{6a1} or \textbf{6b1}. \textbf{6a1} is needed when $d_1
\pgt{p} d_2$, while \textbf{6b1} is needed when $d_2 \pgt{p} d_1$.
\begin{small}
\begin{description}
  \item[\textbf{6a1}] $  
    \matchDvar{d_1}{p} \land \lnot \matchCvar{c}{\rank{c}((p,p))}
    \implies \\ \matchPvar{p}{\rank{p}(d_2)-1}{\quota{p}}$
  \item[\textbf{6b1}] $
    \matchDvar{d_1}{p} \land \lnot \matchCvar{c}{\rank{c}((p,p))} 
    \implies \\ \matchPvar{p}{\rank{p}(d_1)-1}{\quota{p}-1}$
  \item[\textbf{6c}] $
  \begin{array}[t]{l}
    \lnot \matchDvar{d_1}{p} \land \lnot \matchDvar{d_2}{p} \land 
    \lnot \matchCvar{c}{\rank{c}((p,p))}
    \mbox{}\, \implies \\
      \begin{array}[t]{l}
      \matchPvar{p}{\rank{p}(d_1){-}1}{\quota{p}}
      \lor \matchPvar{p}{\rank{p}(d_1){-}1}{\quota{p}{-}1} \\
      \mbox{} \lor \matchPvar{p}{\rank{p}(d_2){-}1}{\quota{p}} \lor
      \matchPvar{p}{\rank{p}(d_2){-}1}{\quota{p}{-}1}
      \end{array}
  \end{array}$
\end{description}
\end{small}
Clauses \textbf{6a1} or \textbf{6b1} say that if $d_1$ is already in
$p$ and $c$ is not matched to $(p,p)$ or into a more preferred program
pair, then $p$ will not accept $d_2$. \textbf{6b1} differs because
when $d_2\pgt{p} d_1$ and $d_1$ is already in $p$, $p$ will definitely
accept $d_2$. In this case, however, the couple is not accepted into
$(p,p)$ if accepting $d_2$ causes $d_1$ to be bumped. That is, when
$\matchPvar{p}{\rank{p}(d_1)-1}{\quota{p}-1}$ is true (adding $d_2$
will cause $\matchPvar{p}{\rank{p}(d_1)-1}{\quota{p}}$ to become
true).

There are also analogous clauses \textbf{6a2} and \textbf{6b2} (one of
which is used) to deal with the case when $d_2$ is already in $p$ and
we need to ensure that $p$ won't accept $d_1$. Clause \textbf{6c}
handles the case when neither member of the couple is currently
matched into $p$.

Let $\prob = \langle \D, \C, \P, \pgeq{} \rangle$ be a matching
problem. We say that a matching $\match$ for $\prob$ and a truth
assignment $\pi$ for SAT-E of $\prob$ are
\textbf{corresponding} when $\pi\models \matchDvar{d}{p}$ iff
$\match(d) = p$.

We provide a full proof of the following theorem in Drummond et al.\ \shortcite{dpb_review}.

\begin{thm}
  If $\match$ and $\pi$ are corresponding, then $\match$ is
  stable if and only if $\pi$ is satisfying.
\end{thm}
This theorem shows that the satisfying assignments of SAT-E
are in a 1-1 relationship with the stable models.


\end{document}